\documentclass[letterpaper, 11 pt]{article} 
\usepackage{amsmath,amsfonts,amssymb,color,amsthm}
\usepackage{nccmath}
\usepackage{epsfig}
\usepackage{psfrag}
\usepackage[T1]{fontenc}
\usepackage{algorithm}
\usepackage{algpseudocode}
\usepackage[dvipsnames]{xcolor}
\usepackage{array}
\usepackage{epstopdf}
\usepackage{cite}
\usepackage{footmisc}
\usepackage{mathtools}
\usepackage{bbm}
\usepackage{bm}
\usepackage{comment}
\usepackage{dsfont}
\usepackage{mathtools}
\usepackage{epstopdf}
\usepackage{tikz}
\usetikzlibrary{automata, positioning, arrows.meta}
\usepackage{relsize}
\usetikzlibrary{shapes,arrows}
\usepackage{cite}
\usepackage{epstopdf}
\usepackage{tcolorbox}
\definecolor{winered}{rgb}{0.5,0,0}
\usepackage{scalerel}
\usepackage{xcolor}
\usepackage[colorlinks]{hyperref}
\usepackage{tcolorbox}
\tcbset{colback=blue!5!white,colframe=blue}

\newcommand{\mc}{\mathcal}

\newcommand{\normlr}[1]{\left\lVert#1\right\rVert}
\newcommand{\stackref}[2]{
\readlist*\mylist{#1}
\stackrel{\mbox{\footnotesize\foreachitem\x\in\mylist[]{\ifnum\xcnt=1\else,\fi\eqref{\x}}}}{#2}
}

\newcommand{\abs}[1]{\lvert #1 \rvert}

\newcommand{\tr}[1]{\text{trace}\left( #1 \right)}

\AtBeginDocument{%
  \hypersetup{
    citecolor=blue,
    linkcolor=winered
  }
}

\DeclarePairedDelimiter\ceil{\lceil}{\rceil}
\DeclareMathOperator*{\argmin}{\arg\!\min}

\DeclarePairedDelimiterX{\norm}[1]{\lVert}{\rVert}{#1}
\DeclareMathOperator{\E}{\mathbb{E}}

\newtheorem{problem}{Problem}

\newtheorem{theorem}{Theorem}

\newtheorem{lemma}{Lemma}

\date{}
\usepackage{cite}
\usepackage[margin=1in]{geometry}

\title{\LARGE 
Boosting-Enabled Robust System Identification of Partially Observed LTI Systems Under Heavy-Tailed Noise
}

\author{Vinay Kanakeri and Aritra Mitra}% <-this % stops a space

\begin{document}
\maketitle
\footnotetext[1]{The authors are with the Department of Electrical and Computer Engineering, North Carolina State University. Email: {\tt \{vkanake, amitra2\}@ncsu.edu}.}
\thispagestyle{empty}
\pagestyle{empty}

%%%%%%%%%%%%%%%%%%%%%%%%%%%%%%%%%%%%%%%%%%%%%%%%%%%%%%%%%%%%%%%%%%%%%%%%%%%%%%%%
\begin{abstract}
We consider the problem of system identification of partially observed linear time-invariant (LTI) systems. Given input-output data, we provide non-asymptotic guarantees for identifying the system parameters under general heavy-tailed noise processes. Unlike previous works that assume Gaussian or sub-Gaussian noise, we consider significantly broader noise distributions that are required to admit only up to the second moment. For this setting, we leverage tools from robust statistics to propose a novel system identification algorithm that exploits the idea of \emph{boosting}. Despite the much weaker noise assumptions, we show that our proposed algorithm achieves sample complexity bounds that nearly match those derived under sub-Gaussian noise. In particular, we establish that our bounds retain a logarithmic dependence on the prescribed failure probability. Interestingly, we show that such bounds can be achieved by requiring just a finite fourth moment on the excitatory input process.
\end{abstract}

%%%%%%%%%%%%%%%%%%%%%%%%%%%%%%%%%%%%%%%%%%%%%%%%%%%%%%%%%%%%%%%%%%%%%%%%%%%%%%%%
\section{Introduction} \label{sec:intro}
System identification is a fundamental problem that involves estimating \emph{unknown} system parameters using noisy data generated from a dynamical process. It is relevant to various disciplines including control theory, economics, time-series forecasting, and machine learning. System identification also forms a core sub-routine in data-driven control/model-based reinforcement learning where one uses the estimated system model for downstream decision-making. To ensure desired performance of such algorithms, it is crucial to quantify the uncertainty in the data-driven estimates of the model. It stands to reason that the nature of such estimates, and the uncertainty intervals around them, will depend on the statistics of the data used for estimation. In this regard, despite the wealth of literature on system identification spanning both classical asymptotic results~\cite{lai1983asymptotic} and more recent finite-time guarantees~\cite{tu2017non, simchowitz1, sarkar2019near}, almost all existing works on the topic crucially rely on the noise processes being either Gaussian or sub-Gaussian, i.e., ``light-tailed". In practice, however, such an idealistic assumption may not hold. Furthermore, estimators that do not account for non-ideal noise processes might lead to poor statistical guarantees that are inadequate for safety-critical real-time feedback control loops. With these points in mind, the \textbf{goal} of this work \emph{is to initiate a study of system identification under more realistic noise processes that are potentially heavy-tailed and admit no more than the second moment.}  

To do so, we ground our study within the following linear time-invariant (LTI) system model:  
\begin{equation}\label{eqn:sys_model}
\begin{split}
x_{t+1} &= A x_t + B u_t + w_t \\
y_t &= C x_t + D u_t + v_t ,
\end{split}
\end{equation}
where $A$, $B$, $C$, and $D$ constitute the unknown parameters of the system, and $\{w_t\}$ and $\{u_t\}$ are the noise sequences. Recent works have considered several variations of the above model and discussed the implications of system stability on the ease of parameter estimation; for a survey of such results, see~\cite{tsiamis2023statistical}. As mentioned earlier, a common theme among previous works is the ideal Gaussian or sub-Gaussian assumption made on the noise sequences. Although such strong assumptions lead to stronger guarantees, it is unclear whether they hold in complex real-world systems. Even if such light-tailed noise assumptions do hold, the fact that they do may not be known a priori.

In light of this fact, we significantly weaken the assumptions on the noise processes. In particular, we consider general heavy-tailed noise processes that admit a finite second moment, and nothing more. For this significantly more general setting, one might ask: \emph{Is it possible to recover similar finite-sample system identification guarantees as under sub-Gaussian noise?} The main contribution of this work is to prove that this is indeed possible via a suitably designed system identification algorithm robust to heavy-tailed noise. 

Similar to \cite{zheng2021nonasymptotic}, we take a two-step approach to estimate the system parameters. Our main result corresponds to learning the first $T$ Markov parameters of the system: 
\begin{equation}
    G = \begin{bmatrix}
        D & CB & CAB & \ldots & CA^{T-2}B 
    \end{bmatrix}. \label{eqn:markov_parameters}
\end{equation}
To estimate the system matrices $A,B,C$, and $D$ from the Markov parameters, one can then readily use the Ho-Kalman algorithm \cite{kalman1966effective}, \cite{oymak2019non}. Our key contributions are as follows. 

$\bullet$ \textbf{Problem Formulation:} Assuming access to multiple independent system trajectories, we study for the first time the problem of non-asymptotic system identification under heavy-tailed noise processes that admit no more than the second moment. For noise processes that admit up to the fourth moment, a recent survey \cite{tsiamis2023statistical} identified the open problem of obtaining the optimal logarithmic dependence on the prescribed confidence level $\delta$. It was also pointed out in this survey that the standard ordinary least-squares (OLS) approach to system identification may fail to achieve such a dependence. In this context, we examine whether a logarithmic dependence can be reinstated under our setting. 

$\bullet$ \textbf{Robust Algorithm:} We device a novel algorithm to estimate the Markov parameters. Our algorithm carefully divides multiple rollouts into different buckets, and combines weak estimators obtained from each bucket to output a strongly concentrated estimate of the Markov parameters. To ``boost" the performance of the  weak estimators, we use the geometric median with respect to the Frobenius norm \cite{minsker2015geometric}.    

$\bullet$ \textbf{Finite Sample Analysis:} Our main result, namely Theorem~\ref{thm:markov_parameters}, reveals that our proposed approach is able to nearly recover the same system identification guarantees as achieved in prior work~\cite{zheng2021nonasymptotic} under the stronger assumptions of sub-Gaussian noise. In particular, our bounds recover the desired logarithmic dependence on the failure probability. As far as we are aware, \emph{this is the first work to establish such a result for partially observable LTI systems} of the form in~\eqref{eqn:sys_model}. 

In terms of the analysis, existing works use Chernoff-type bounds from high-dimensional statistics~\cite{vershynin2018high} available for sub-Gaussian and sub-exponential random variables. In the absence of such bounds for our heavy-tailed setting, we show how basic concentration tools can be combined with ideas from robust statistics to yield strong estimation guarantees. The resulting proof ends up being quite simple and self-contained. While our algorithmic and analysis ideas are inspired by similar concepts explored in robust mean estimation\cite{lugosi2021robust}, our contribution lies in revealing how such concepts can be exploited for system identification as well. 

Overall, with the above contributions, our work refines the current understanding of non-asymptotic system identification by shedding light on the extent to which prevalent idealistic noise assumptions can be relaxed, without compromising on the final performance guarantees. We hope that the use of high-dimensional robust statistics tools will pave the way for reasoning about complex data-driven control problems under non-ideal assumptions on the data.

\textbf{Related Work.} For a classical treatment of the topic of system identification, we refer the reader to~\cite{lai1983asymptotic, ljung1998system}. Given the widespread applications of data-driven control, there has been a resurgence of interest in system identification over the last few years. In contrast to the asymptotic analysis pursued in more classical works, the more recent literature has focused on providing a fine-grained finite-sample complexity analysis of popular system identification algorithms, by drawing on tools from concentration of measure and statistical learning theory. One way to categorize this literature is by distinguishing the \emph{multi-trajectory} setting that we consider here from the relatively more challenging \emph{single-trajectory} setting. In the former setting considered in~\cite{dean2020sample, zheng2021nonasymptotic, sun2020finite, xing2020linear, xin2022learning, 
tu2024learning}, the learner has access to multiple independent rollouts of the system trajectories; the statistical independence across trajectories simplifies the analysis to some extent. In the single-rollout setting without resets, the difficulty stems from having to contend with temporal correlations in the data. Such a setting has been investigated for both linear~\cite{rantzer2018, simchowitz1, tu2017non, oymak2019non, tsiamis2019finite, jedra2022finite} and nonlinear systems~\cite{foster2020learning, ziemann2022single}. Yet another way to classify the system identification literature is to consider whether the state can be directly observed (fully observed case) or whether it is only partially observable via noisy measurements. In this paper, we consider the latter case which has been studied in~\cite{oymak2019non, sun2020finite, tsiamis2019finite, zheng2021nonasymptotic}. 

A common unifying theme across all the papers reviewed above is that the noise processes are assumed to be ``light-tailed", i.e., sub-Gaussian or Gaussian. A notable exception is~\cite{faradonbeh2018finite}, where the authors consider a sub-Weibull noise process. Nonetheless, like sub-Gaussian distributions, sub-Weibull distributions admit all finite moments, as shown in~\cite{vladimirova2020sub}. In sharp contrast, our work assumes nothing more than the existence of the second moment of the noise processes, and fourth moment of the excitation input process. 

In our very recent work \cite{kanakeri2024outlier}, we considered the problem of estimating the state transition matrix of a fully observable LTI system excited by a heavy-tailed process noise that admits up to the fourth moment. The key departures of this paper from~\cite{kanakeri2024outlier} are as follows: (i) we consider a partially observable setting, (ii) there is no measurement noise or inputs in~\cite{kanakeri2024outlier}, and (iii) we further relax the noise assumptions relative to~\cite{kanakeri2024outlier}. In particular, \cite{kanakeri2024outlier} assumes that the noise process admits a finite fourth moment. We show that access to excitation inputs that admit up to the fourth moment allows for relaxing the assumptions on the noise processes.  

\textbf{Notation:} Given a positive integer $n \in \mathbb{N}$, we define the shorthand $[n] \triangleq \{1, 2, \ldots, n\}.$ For a vector $w \in \mathbb{R}^d$, we will use $w^{\top}$ to denote its transpose, and $w(i)$ to represent its $i$-th component. Unless otherwise specified, $\Vert \cdot \Vert$ will be used to denote the Euclidean norm for vectors and spectral norm for matrices. Given a matrix $M \in \mathbb{R}^{m \times n}$, we will use $\Vert M \Vert_F$ to denote its Frobenius norm. We use $I_d$ to represent the identity matrix in $\mathbb{R}^{d \times d}$. Given matrices $A$ and $B$ in $\mathbb{R}^{d \times d}$, $A \succeq B$ means $A - B$ is positive semi-definite. Finally, we will use $c, c_1, c_2, \ldots$ to represent universal constants that may change from one line to another.

\section{Problem Formulation}
\label{sec:prob_form}
Consider a discrete-time LTI system as in \eqref{eqn:sys_model} where at time $t$, $x_t \in \mathbb{R}^n$ is the state of the system, $y_t \in \mathbb{R}^p$ is the output, $u_t \in \mathbb{R}^m$ is the input, $w_t \in \mathbb{R}^n$ is the process noise, and $v_t \in \mathbb{R}^p$ is the measurement noise. The matrices $A \in \mathbb{R}^{n \times n}$, $B \in \mathbb{R}^{n \times m}$, $C \in \mathbb{R}^{p \times n}$, and $D \in \mathbb{R}^{p \times m}$ are a priori \emph{unknown}, and the goal of system identification is to estimate them using the input and output sequences. Before formalizing the problem statement, we provide the assumptions pertaining to the system model. 

\textbf{Assumption 1.} The noise sequence $\{w_t\}$ and $\{v_t\}$ are mutually independent stochastic processes that are assumed to be heavy tailed, with the existence of \emph{only the first and the second moment}. More precisely, we assume that $\{w_t\}$ and $\{v_t\}$ are both zero mean independent and identically distributed (i.i.d) stochastic processes with the following covariance matrices $\forall t \geq 0$:
\begin{equation} \label{eqn:noise_model}
    \mathbb{E}[w_t w^{\top}_t] = \sigma_w^2 I_n, \ \mathbb{E}[v_t v^{\top}_t] = \sigma_v^2 I_p.
\end{equation}

\textbf{Assumption 2.} We assume that the system starts from a zero initial state, i.e., $x_0 = 0$. This is a standard assumption made in the literature as it helps in isolating the challenges posed by the noise processes in the analysis \cite{zheng2021nonasymptotic}. Moreover, we expect that handling a non-zero initial condition would require only a straightforward extension of our approach.

\subsection{Data collection}
Multiple rollouts of input and output sequences from the dynamical system \eqref{eqn:sys_model} are collected. Each rollout is obtained by starting the system from the zero initial state and letting it evolve based on the inputs provided for $T$ time steps ($0$ to $T-1$). We denote the data comprising of the output and input sequences from the $i$-th rollout by
\begin{equation}
    \mc{D}^{(i)} = \{(y_t^{(i)}, u_t^{(i)}): 0\leq t \leq T-1\}. \nonumber
\end{equation}
The collective dataset for $N$ such rollouts is denoted by $\mc{D} = \bigcup_{i \in [N]} \mc{D}^{(i)}$. Since the source of excitation is the input sequence, a natural question is to understand the weakest assumptions on the excitation process to achieve persistence of excitation. Hence, although previous works have considered Gaussian input sequence, we choose a stochastic process that resembles the general heavy-tailed noise model in \eqref{eqn:noise_model} with an additional requirement that it has a finite fourth moment. In particular, our input sequence is a zero mean i.i.d stochastic process with the following second and fourth moment $\forall t \geq 0$:
\begin{equation}
    \mathbb{E}[u_t u^{\top}_t] = \sigma_u^2 I_m, \ \mathbb{E}[(u_t(i))^4] = \tilde{\sigma}_u^4, \forall i \in [m]. \label{eqn:input_model}
\end{equation} 
Furthermore, we assume that the input sequence, $\{u_t\}$, is mutually independent of both the noise sequences, $\{w_t\}$ and $\{u_t\}$.
We now formalize the problem statement.

% \textbf{Problem 1:} Consider the system in \eqref{eqn:sys_model} with the noise sequences adhering to \eqref{eqn:noise_model}. Given the dataset $\mc{D}$, an accuracy parameter $\epsilon > 0$, and a failure probability $\delta \in (0, 1)$, construct estimators $\hat{A}, \hat{B}, \hat{C}, $ and $\hat{D}$ and characterize a sample complexity $N_S(\epsilon, \delta, C_m, C_n)$ such that with probability at least $1 - \delta$, $\max \{\norm{A - S\hat{A}S^\top}, \norm{B - S\hat{B}}, \norm{C - \hat{C}S^\top}, \norm{D - \hat{D}}\} \leq \epsilon$ provided $N \geq N_S(\epsilon, \delta, C_m, C_n)$, where $C_m$ depends on system parameters and dimensions, and $C_n$ depends on the noise parameters. 

\begin{problem} {\textbf{(Robust System Identification):}} \label{pblm:robust_sysID}
Consider the dataset $\mc{D}$ consisting of $N$ rollouts, each of length $T$, from the system described in \eqref{eqn:sys_model}, where the noise sequences follow \eqref{eqn:noise_model}. Given a failure probability $\delta \in (0, 1)$, construct an estimator $\hat{G}$ of the first $T$ Markov parameters denoted as $G$ in~\eqref{eqn:markov_parameters}, and establish an upper bound $\epsilon(\delta, N)$ such that
\begin{equation*} 
\begin{aligned}
    \norm{\hat{G} - G} \leq \epsilon(\delta, N)
\end{aligned}
\end{equation*}
provided $N \geq N_S(\delta)$. The sample complexity $N_S$ and the upper bound $\epsilon$ may also depend on the system parameters, noise parameters, and the dimensions. 
\end{problem} 

The following remarks highlight the key departures of the above problem from previous works on system identification. 
\begin{enumerate}
    \item An important distinction of this work lies in the generality of our assumptions on the noise sequences $\{w_t\}$ and $\{v_t\}$. We consider the noise sequences that have only a finite second moment, encapsulating a broad class of distributions. In sharp contrast, previous works have studied the problem under the Gaussian noise sequences \cite{dean2020sample}, \cite{zheng2021nonasymptotic}, \cite{tsiamis2019finite}. Notably, it is well known that all finite moments of the Gaussian and sub-Gaussian distributions exist - see proposition 2.5.2 of \cite{vershynin2018high}.

    \item In the Gaussian noise setting, as shown in the previous works \cite{dean2020sample}, \cite{zheng2021nonasymptotic}, one can achieve a logarithmic dependence on the failure probability $\delta$ for sample complexity and error bounds. This is made possible by leveraging stronger Chernoff-like bounds, which provide exponentially decreasing upper bounds on tail probabilities. However, these methods rely on the existence of a moment-generating function, which, in turn, implies the existence of all finite moments. In this regard, an important question is whether one can retain the logarithmic dependence on the failure probability under our significantly weaker assumption on the noise process. We show that it is indeed possible by a suitably designed algorithm.  
    
    \item Our previous work \cite{kanakeri2024outlier}, also focuses on the problem of system identification under heavy-tailed noise. However, that setting considers fully observed systems, without measurement noise or inputs. Additionally, the noise sequence in \cite{kanakeri2024outlier} is assumed to have moments up to the fourth order, similar to the input sequence \eqref{eqn:input_model} considered in this work. \emph{Our analysis reveals that it suffices to ensure the existence of fourth moment for the sequence that excites the system - whether it be the noise sequence in \cite{kanakeri2024outlier} or input sequence in this work.} 
\end{enumerate}

We present our algorithm in the following section.
\section{Robust System Identification Algorithm}
\label{sec:algorithm}
The algorithm we propose involves three main steps. We describe these steps in detail below.

\textbf{Step 1: Bucketing.} Partition the dataset $\mc{D}$ comprising of $N$ rollouts into $K$ buckets each containing $M$ independent rollouts. We denote the $K$ buckets by $\mc{B}_1, \ldots, \mc{B}_K$. 

\textbf{Step 2: Least squares estimation per bucket.} Construct OLS estimators of the Markov parameters for each bucket. We denote the estimator for the $j$-th bucket by $\hat{G}_j$. Later in this section, we formulate the least squares problem to estimate the Markov parameters and provide an expression for $\hat{G}_j$.

\textbf{Step 3: Boosting.} Boost the performance of the weakly concentrated OLS estimators from step 2 by fusing them together using the geometric median with respect to the Frobenius norm. More precisely, we construct an estimator $\hat{G}$ as follows:
\begin{equation}
    \hat{G} = \texttt{Med}(\hat{G}_1, \ldots, \hat{G}_K) = \argmin_{\theta \in \mathbb{R}^{p \times mT}} \sum_{j \in [K]} \norm{\theta - \hat{G}_j}_F. \label{eqn:geo_median}
\end{equation}

The following comments concerning our algorithm are now in order:
\begin{enumerate}
    \item In the step 1, we have assumed that $N = MK$ for simplicity. The choice of the number of buckets, $K$, is a design parameter in our algorithm, and our analysis naturally provides the value of $K$ that meets the confidence level $\delta$. Lemma \ref{lemma:boosting} in Section \ref{sec:analysis} elaborates on this. 
    \item The boosting step (Step 3) is motivated by the ``median of means" device from robust statistics. This notion is extended to a general Hilbert space using the geometric median in \cite{minsker2015geometric}. In our setting, we compute the geometric median with respect to the Frobenius norm\footnote{The Frobenius inner-product of two matrices $X$ and $Y$ in $\mathbb{R}^{m \times n}$ is defined as $\tr{X^\top Y}$. The Frobenius norm of $X \in \mathbb{R}^{m \times n}$ is defined as $\sqrt{\tr{X^\top X}}$.} as it induces an inner-product space on the space of all matrices in $\mathbb{R}^{p \times mT}$.
\end{enumerate}
In the following, we present a least squares formulation, similar to that in \cite{zheng2021nonasymptotic}, to estimate the Markov parameters for each bucket. 

\subsection{Least squares estimators}
Based on the model \eqref{eqn:sys_model} and the zero initial state assumption, we have $x_t^{(i)} = \sum_{k = 0}^{t - 1} A^k\left( 
Bu_{t-k-1}^{(i)} + w_{t-k-1}^{(i)} \right),$ which leads to $$y_{t}^{(i)} = \sum_{k = 0}^{t - 1} CA^k\left( 
Bu_{t-k-1}^{(i)} + w_{t-k-1}^{(i)}\right) + Du_t^{(i)} + v_t^{(i)}.$$ Defining $y^{(i)}$ as
$$ y^{(i)} = \begin{bmatrix} y_0^{(i)} & y_1^{(i)} & \dots & y_{T-1}^{(i)} \end{bmatrix} \in \mathbb{R}^{p \times T},$$ and $U^{(i)}$ as
\begin{align*}
    U^{(i)} &=
\begin{bmatrix}
u_0^{(i)} & u_1^{(i)} & u_2^{(i)} & \dots & u_{T-1}^{(i)} \\
0 & u_0^{(i)} & u_1^{(i)} & \dots & u_{T-2}^{(i)} \\
0 & 0 & u_0^{(i)} & \dots & u_{T-3}^{(i)} \\
\vdots & \vdots & \vdots & \ddots & \vdots \\
0 & 0 & 0 & \dots & u_0^{(i)}
\end{bmatrix} \in \mathbb{R}^{(mT \times T)},
\end{align*}
we can compactly represent the data corresponding to the $i$-th rollout as $$y^{(i)} = GU^{(i)} + FW^{(i)} + v^{(i)},$$ where $W^{(i)}$ and $v^{(i)}$ are defined similar to $U^{(i)}$ and $y^{(i)}$, respectively, $G \in \mathbb{R}^{p \times mT}$ is defined in \eqref{eqn:markov_parameters}, and $F = \begin{bmatrix} 0 & C & CA & \ldots & CA^{T-2} \end{bmatrix}\in \mathbb{R}^{p \times nT}.$ The OLS problem to estimate the Markov parameters corresponding to the $j$-th bucket can be posed as 
\begin{equation*}
    \hat{G}_j = \argmin_{\theta \in \mathbb{R}^{p\times mT}} \sum_{i \in \mc{B}_j}\norm{y^{(i)} - \theta U^{(i)}}_F^2.
\end{equation*}
Defining the data corresponding to the $j$-th bucket as 
\begin{align}
    Y_j &= \begin{bmatrix}
        y^{(1_j)} & y^{(2_j)} & \ldots & y^{(M_j)}
    \end{bmatrix} \in \mathbb{R}^{p \times MT}, \label{eqn:Y_j_def} \\
    U_j &= \begin{bmatrix}
        U^{(1_j)} & U^{(2_j)} & \ldots & U^{(M_j)}
    \end{bmatrix} \in \mathbb{R}^{mT \times MT}, \label{eqn:U_j_def}
\end{align}
where the indexing $l_j$ corresponds to the $l$-th trajectory in the $j$-th bucket, the OLS formulation can now be represented as $$\hat{G}_j = \argmin_{\theta\in \mathbb{R}^{p\times mT}} \norm{Y_j - \theta U_j}_F^2.$$ The least squares solution is $\hat{G}_j = Y_jU_j^\dagger$, where $U_j^\dagger = U_j^\top(U_jU_j^\top)^{-1}$ is the right pseudo-inverse of $U_j$. These least squares estimators from all the buckets are then used in step 3, where we compute the geometric median with respect to the Frobenius norm \eqref{eqn:geo_median}. 

\section{Finite Sample Analysis} \label{sec:analysis}
In this section, we provide a finite sample analysis of the system identification algorithm presented in section \ref{sec:algorithm}. We first provide high probability guarantees for the Markov parameters estimator, and then, translate the bound to the system parameters obtained by using the Ho-Kalman algorithm. The following theorem captures our main result concerning the recovery of Markov parameters.

\begin{theorem}{\textbf{(Markov parameters recovery):}} \label{thm:markov_parameters}
Consider the system in \eqref{eqn:sys_model} and the noise assumptions in \eqref{eqn:noise_model}. Fixing $\delta \in (0, 1)$, with probability at least $1 - \delta$, the following holds for the Markov parameters estimated by the algorithm described in Section \ref{sec:algorithm}:
\begin{equation} \label{eqn:main_result}
    \norm{\hat{G} - G} \leq \left(\frac{\sigma_vC_1 + \sigma_wC_2}{\sigma_u}\right) \sqrt{\frac{p\log(1/\delta)}{N}}, 
\end{equation}
where 
\begin{equation} \label{eqn:c1_c2}
    \begin{split}
        C_1 &= c_1 T^{1.5}\sqrt{pm}, \\      
        C_2 &= c_2\norm{F}T^{2.5}\sqrt{nm},
    \end{split}
\end{equation}
provided $K = \ceil{32 \log(1/\delta)}$, $M \geq c_4 (mT)^2(\tilde{\sigma}_u^4/\sigma_u^4)$, and $N = MK$.
\end{theorem}

\textbf{Remark 1. (Comparison with previous work):} We compare our results with those of \cite{zheng2021nonasymptotic}, where the authors studied the problem under Gaussian noise and input processes. Our result shows that we can successfully recover the logarithmic dependence on the failure probability $\delta$ even in the presence of general heavy-tailed noise processes. To our knowledge, \emph{this is the first result to provide such a guarantee for a partially observed system with measurement noise.} However, in terms of dimensions, we incur an extra multiplicative factor of $T\max(n, m, p)$ in our error bound and the number of rollouts as compared to Theorem 1 in \cite{zheng2021nonasymptotic}. That said, although stated for the spectral norm, our error bound also holds for the Frobenius norm in Theorem 1. Furthermore, as discussed later in this section, the extra multiplicative factor in the number of rollouts can be avoided by using a Gaussian input sequence. We elaborate the causes of additional dimension factors later in this section.

In the following, we provide a proof sketch of Theorem \ref{thm:markov_parameters}. The structure of our analysis resembles that presented in \cite{zheng2021nonasymptotic}. However, major departures arise while bounding the norms of random matrices. We appeal to the Markov's inequality to bound the general heavy-tailed noise instead of the well-known concentration tools that are only available for Gaussian and sub-Gaussian matrices.

\textbf{Proof sketch for Theorem \ref{thm:markov_parameters}:} 
We start by deriving the bounds for the OLS estimators of Markov parameters from each bucket. Fixing the bucket $j$, notice that the error $\hat{G}_j - G$ can be expressed as $\hat{G}_j - G = (FW_j + V_j)U_j^\top(U_jU_j^\top)^{-1}$, where $W_j \in \mathbb{R}^{nT \times MT}$ is defined similar to $U_j$ as in \eqref{eqn:U_j_def} and $V_j \in \mathbb{R}^{p \times MT}$ similar to $Y_j$ as in \eqref{eqn:Y_j_def}. By submultiplicativity of the spectral norm, we have 
\begin{equation}
    \norm{\hat{G}_j - G} \leq \norm{(U_jU_j^\top)^{-1}}(\norm{F}\norm{W_jU_j^\top} + \norm{V_jU_j^\top}). \label{eqn:g_error_sub_mul}
\end{equation}
Therefore, to obtain an upper bound on $\norm{\hat{G}_j - G}$, it suffices to individually upper bound the terms $\norm{(U_jU_j^\top)^{-1}}$, $\norm{W_jU_j^\top}$, and $\norm{V_jU_j^\top}$. In this regard, the following three lemmas, with proofs provided in Appendix \ref{app:proofs}, capture the key results.  

\begin{lemma}{\textbf{(Invertibility of $U_jU_j^{\top}$):}}\label{lemma:invertibility}
Fix $q \in (0, 1)$. With probability at least $1 - q/3$, $\lambda_{\min}(U_jU_j^\top) \geq M\sigma_u^2/2$, provided $M \geq c_1(1/q)(mT)^2 (\tilde{\sigma}_u^4/\sigma_u^4)$.
\end{lemma}

\begin{lemma}{\textbf{(Bounding process noise):}}\label{lemma:process_noise}
Fix $q \in (0, 1)$. With probability at least $1 - q/3$, the following holds:
\begin{equation}
    \norm{W_jU_j^\top} \leq c_2 T^{2.5} \sqrt{\frac{M n m}{q}}\sigma_u \sigma_w. \label{eqn:num_process_noise}
\end{equation}
\end{lemma}

\begin{lemma}{\textbf{(Bounding measurement noise):}}\label{lemma:measurement_noise}
Fix $q \in (0, 1)$. With probability at least $1 - q/3$, the following holds:
\begin{equation}
    \norm{V_jU_j^\top} \leq c_1 T^{1.5} \sqrt{\frac{Mpm}{q}}\sigma_u\sigma_v.
\end{equation}
\end{lemma}

Equipped with the results from the above three lemmas, we can bound the term $\norm{\hat{G}_j - G}$, by applying an union bound. Hence, with probability at least $1 - q$, we have 
\begin{equation}
    \norm{\hat{G}_j -G} \leq \frac{1}{\sqrt{Mq}} \frac{\sigma_vC_1 + \sigma_wC_2}{\sigma_u}, \label{eqn:bound_before_boosting}
\end{equation}
provided $M \geq c_1(1/q)(mT)^2 (\tilde{\sigma}_u^4/\sigma_u^4)$. 

\textbf{Remark 2. (Role of the input sequence):} Notice that the number of rollouts in each bucket, $M$, is controlled entirely by Lemma \ref{lemma:invertibility}, which concerns the smallest eigenvalue of $U_jU_j^\top$. Interestingly, our analysis reveals that $M$ depends on the kurtosis of the heavy-tailed input distribution - defined as the ratio of fourth moment to the square of the variance. A similar dependence also appears in \cite{kanakeri2024outlier} where we studied a fully observed system without input. In that case, since the system was excited solely by the noise process, the sample complexity depended on the kurtosis of the noise distribution. Additionally, the proof of Lemma \ref{lemma:invertibility} in Appendix \ref{app:proofs} demonstrates the crucial role of fourth moment in ensuring the invertibility of $U_jU_j^{\top}$. Overall, together with the results from \cite{kanakeri2024outlier}, it is clear that the existence of a finite fourth moment on the sequence exciting the system is sufficient to ensure invertibility with high probability. Whether it is also necessary remains an open question.

An alternative approach is to excite the system with a Gaussian input sequence like in \cite{zheng2021nonasymptotic}. In this case, the number of samples per bucket $M$, does not depend on kurtosis, and is in the order of $mT$ with a logarithmic dependence on the failure probability $q/3$.

In our current setting, the error bound in \eqref{eqn:bound_before_boosting}, which corresponds to the OLS estimators for each bucket from step 2 of our algorithm, has a weaker polynomial dependence on the failure probability $q$. Notice that even if the input sequence were Gaussian or sub-Gaussian, a standard OLS estimator alone would not be sufficient to achieve the $\log(1/\delta)$ dependence. This is clear based on Lemmas \ref{lemma:measurement_noise} and \ref{lemma:process_noise} which have error bounds with a $\sqrt{1/q}$ dependence due to the heavy-tailed noise processes. Hence, to boost the performance of these estimators, we proceed to step 3 of the algorithm where we compute the geometric median of the weakly concentrated estimators. The following lemma captures the result corresponding to the boosting step.

\begin{lemma}{\textbf{(Boosting):}} \label{lemma:boosting}
Fix $q = 1/8$ in \eqref{eqn:bound_before_boosting}. Let $\hat{G} = \texttt{Med}(\hat{G}_1, \ldots, \hat{G}_K)$, where $\hat{G}_1, \ldots, \hat{G}_K$ are $K$ independent OLS estimators corresponding to different buckets, each satisfying \eqref{eqn:bound_before_boosting}. Given $\delta \in (0, 1)$, when $K \geq \ceil{32 \log(1/\delta)}$, the following holds with probability at least $1 - \delta$:
\begin{equation}
    \norm{\hat{G} - G} \leq \left(\frac{\sigma_vC_1 + \sigma_wC_2}{\sigma_u}\right) \sqrt{\frac{p\log(1/\delta)}{N}}. \label{eqn:boosting_lemma}
\end{equation}
\end{lemma}
\begin{proof}
    We leverage Lemma \ref{lemma:geometric_median} from Appendix \ref{app:preliminary_results}, adapted from Lemma 2.1 in \cite{minsker2015geometric}, as a key component of the proof. Let $\varepsilon = (\sqrt{8/M}) (1/\sigma_u)(\sigma_vC_1 +  \sigma_wC_2)$, and set $\alpha = 1/4$, which makes $C_\alpha = 3/(2\sqrt{2})$. Now, we define a ``bad" event as the geometric median of the OLS estimators, $\hat{G}$, being far from the true Markov parameter $G$ as in the statement of Lemma \ref{lemma:geometric_median}. More precisely, the ``bad" event indicates $\norm{\hat{G} - G} > C_\alpha\sqrt{p}\varepsilon$, where $p$ is the dimension of the measurements from \eqref{eqn:sys_model}. Furthermore, we define indicator random variables of the events $\norm{\hat{G}_j - G}_F \geq \sqrt{p}\varepsilon$ and $\norm{\hat{G}_j - G} \geq \varepsilon$ as $S_j$ and $Z_j$, respectively, for all $j \in [K]$. Based on these definitions, we have
    \begin{align*}
        \norm{\hat{G} - G} > C_\alpha\sqrt{p}\varepsilon &\overset{(a)}{\implies} \norm{\hat{G} - G}_F > C_\alpha\sqrt{p}\varepsilon \\
        &\overset{(b)}{\implies} \sum_{j \in [K]} S_j > K/4 \\
        &\overset{(c)}{\implies} \sum_{j \in [K]} Z_j > K/4.
    \end{align*}
    In the above steps, $(a)$ follows from the fact that $\norm{\hat{G} - G} \leq \norm{\hat{G} - G}_F$, $(b)$ due to Lemma \ref{lemma:geometric_median} since we have defined the indicator random variables $S_j$ appropriately, and $(c)$ follows from the fact\footnote{For a matrix $A \in \mathbb{R}^{m \times n}$ with singular values $\{\sigma_1 \geq \ldots \geq \sigma_{r(A)} \}$ where $r(A)$ is the rank of $A$, $\norm{A}_F = \sqrt{\sigma_1^2 + \ldots + \sigma_{r(A)}^2}$ and $\norm{A} = \sigma_1(A)$. Therefore, we have $\norm{A} \leq \norm{A}_F \leq \min(m, n)\norm{A}$.} that $\norm{\hat{G}_j - G}_F \leq \sqrt{p} \norm{\hat{G}_j - G}$. Based on the above implications, we have
    \begin{align}
        \mathbb{P}\left( \norm{\hat{G} - G} > C_\alpha \sqrt{p}\varepsilon \right) &\leq \mathbb{P} \left( \sum_{j \in [K]} Z_j \geq K/4 \right). \label{eqn:boosting_proof_1}
    \end{align}
    Since each of the $Z_j$'s are i.i.d. random variables in $\{0,1\}$, we use Hoeffding's inequality to infer that
    \begin{equation*}
        \begin{aligned}
        \mathbb{P} \left( \sum_{j \in [K]} Z_j \geq K/4 \right) &\overset{(a)}{=} \mathbb{P} \left( \frac{1}{K} \sum_{j \in [K]} (Z_j - \E[Z_1]) \geq \frac{1}{4} - \E[Z_1] \right) \\
        &\overset{(b)}{\leq} \exp{\left(-2 K (1/4 - \E[Z_1])^2\right)} \\
        & \overset{(c)}{\leq} \exp{\left(-2 K (1/4 - 1/8)^2\right)},
    \end{aligned}
    \end{equation*}
    where $(a)$ follows by subtracting the common mean as the rollouts are identically distributed, $(b)$ due to Hoeffding's inequality, and $(c)$ follows due to $\E[Z_1] = \mathbb{P}(\norm{\hat{G}_1 - G} \geq \varepsilon) \leq q = 1/8.$ Using the above bound in \eqref{eqn:boosting_proof_1}, we have
    $$ \mathbb{P}\left(\norm{\hat{G} - G} > C_\alpha \sqrt{p}\varepsilon \right) \leq \exp(-K/32) \leq \delta,$$
    when $K=\ceil{32 \log(1/\delta)}$. Therefore, with probability at least $1 - \delta$, we have $\norm{\hat{G} - G} \leq C_\alpha \sqrt{p}\varepsilon$. Using $M = N/K$ in $\varepsilon$, completes the proof.
\end{proof}

\textbf{Remark 3. (Importance of boosting):} From the proof of Lemma \ref{lemma:boosting}, it is clear that the boosting step is crucial to get a logarithmic dependence on the failure probability. This was made possible by the robustness of the geometric median. To see this, consider a scalar case where the geometric median is simply the scalar median. For the median $\hat{G}$ to deviate from $G$ beyond a tolerance $\varepsilon$, at least half of the OLS estimates from the $K$ buckets must also deviate beyond $\varepsilon$. Although the probability of an individual estimate deviating could be large (at most $1/8$ in the proof), half of them \emph{deviating simultaneously diminishes the overall failure probability}. To extend this intuition to the vector case, Lemma \ref{lemma:geometric_median} as adapted from Lemma 2.1 in \cite{minsker2015geometric} uses the geometric median with respect to the Frobenius norm as shown in \eqref{eqn:geo_median}. 

Observe that Lemma \ref{lemma:boosting} also informs the choice of the number of buckets, $K$, in our algorithm. In particular, $K$ depends only of the confidence level $\delta$ as $K = \ceil{32 \log(1/\delta)}$.

\textbf{Remark 4. (Additional dimensional factors):} As mentioned earlier, we incur an extra multiplicative factor of $T \max(n, m, p)$ in our bounds compared to that in \cite{zheng2021nonasymptotic}. This is because the heavy-tailed noise under our study precludes the use of sub-Gaussian and sub-exponential concentration tools with logarithmic dependence on the error probability. The logarithmic factor controls the exponential dependence on dimensions that appear while applying union bounds, a property widely used in previous works with Gaussian /sub-Gaussian noise \cite{zheng2021nonasymptotic}, \cite{dean2020sample}.   

On the other hand, due to unavailability of such concentration tools for general heavy-tailed noise processes, we use Markov's inequality in the proofs of Lemma \ref{lemma:invertibility}-\ref{lemma:measurement_noise}, causing the additional dimension terms in our bounds. More precisely, in the proof of Lemma \ref{lemma:process_noise}, we incur a factor of $T$ as we bound the block matrices and use union bound to transfer the bounds to the complete matrix which is composed on $T^2$ block matrices. Furthermore, we also incur a dimensional factor as our bound involves computing expectation of Frobenius norms of rank 1 random matrices. For instance, in the proof of Lemma \ref{lemma:measurement_noise}, using Markov's inequality leads to $\E[\norm{v_0^{(1)}(u_0^{(1)})^\top}_F^2] = \sigma_v^2\sigma_u^2pm$. On the contrary, using sub-Gaussian tail bounds leads to a factor $(m+p)$ as shown in \cite{zheng2021nonasymptotic}.  

Finally, one can obtain the estimates of $(A, B, C, D)$ up to a similarity transformation by using the Ho-Kalman algorithm on $\hat{G}$, the estimated Markov parameters from step 3. Based on the robustness analysis of the Ho-Kalman algorithm as presented in \cite{oymak2019non}, the bounds on the system parameters are upper bounded by $\norm{\hat{G} - G}$ scaled by system and dimension-dependent constants. As a result, the error bounds on the system parameters retain the $\log(1/\delta)$ dependence on the failure probability from Theorem~\ref{thm:markov_parameters}. 

\section{Conclusion and Future Work}
In this work, we considered the problem of system identification for partially observed LTI systems with heavy-tailed noise processes that admit only the second moment. We showed, by constructing an algorithm and the subsequent finite sample analysis, that one can achieve sample complexity bounds with a logarithmic dependence on failure probability. We also showed that our bounds almost match the sub-Gaussian bounds with extra dimensional factors. 

In the future, it would be of great interest to understand the tightness of the dimension dependencies by deriving information-theoretic lower bounds for the heavy-tailed noise settings. Furthermore, we would like to derive lower bounds taking into account the hardness of learning the system parameters as shown in \cite{tsiamis2021linear}. Another interesting direction would be to provide sample complexity results for system identification based on data from a single rollout. This problem is significantly harder owing to temporally correlated observations. Finally, whether the ideas used in this paper can be extended to heavy-tailed non-linear system identification remains an open problem. We hope that this work motivates further exploration of system identification under non-ideal assumptions on the noise and excitation processes.  

\bibliographystyle{unsrt}
\bibliography{references.bib}

\begin{thebibliography}{10}

\bibitem{lai1983asymptotic}
Tze~Leung Lai and Ching~Zong Wei.
\newblock Asymptotic properties of general autoregressive models and strong consistency of least-squares estimates of their parameters.
\newblock {\em Journal of multivariate analysis}, 13(1):1--23, 1983.

\bibitem{tu2017non}
Stephen Tu, Ross Boczar, Andrew Packard, and Benjamin Recht.
\newblock Non-asymptotic analysis of robust control from coarse-grained identification.
\newblock {\em arXiv preprint arXiv:1707.04791}, 2017.

\bibitem{simchowitz1}
Max Simchowitz, Horia Mania, Stephen Tu, Michael~I Jordan, and Benjamin Recht.
\newblock Learning without mixing: Towards a sharp analysis of linear system identification.
\newblock In {\em Conference On Learning Theory}, pages 439--473. PMLR, 2018.

\bibitem{sarkar2019near}
Tuhin Sarkar and Alexander Rakhlin.
\newblock Near optimal finite time identification of arbitrary linear dynamical systems.
\newblock In {\em International Conference on Machine Learning}, pages 5610--5618. PMLR, 2019.

\bibitem{tsiamis2023statistical}
Anastasios Tsiamis, Ingvar Ziemann, Nikolai Matni, and George~J Pappas.
\newblock Statistical learning theory for control: A finite-sample perspective.
\newblock {\em IEEE Control Systems Magazine}, 43(6):67--97, 2023.

\bibitem{zheng2021nonasymptotic}
Yang Zheng and Na~Li.
\newblock Non-asymptotic identification of linear dynamical systems using multiple trajectories.
\newblock {\em IEEE Control Systems Letters}, 5(5):1693--1698, 2020.

\bibitem{kalman1966effective}
RE~Kalman.
\newblock Effective construction of linear state-variable models from input/output functions.
\newblock {\em at-Automatisierungstechnik}, 14(1-12):545--548, 1966.

\bibitem{oymak2019non}
Samet Oymak and Necmiye Ozay.
\newblock Non-asymptotic identification of lti systems from a single trajectory.
\newblock In {\em 2019 American control conference (ACC)}, pages 5655--5661. IEEE, 2019.

\bibitem{minsker2015geometric}
Stanislav Minsker.
\newblock Geometric median and robust estimation in banach spaces.
\newblock 2015.

\bibitem{vershynin2018high}
Roman Vershynin.
\newblock {\em High-dimensional probability: An introduction with applications in data science}, volume~47.
\newblock Cambridge university press, 2018.

\bibitem{lugosi2021robust}
Gabor Lugosi and Shahar Mendelson.
\newblock Robust multivariate mean estimation: the optimality of trimmed mean.
\newblock 2021.

\bibitem{ljung1998system}
Lennart Ljung.
\newblock System identification.
\newblock In {\em Signal analysis and prediction}, pages 163--173. Springer, 1998.

\bibitem{dean2020sample}
Sarah Dean, Horia Mania, Nikolai Matni, Benjamin Recht, and Stephen Tu.
\newblock On the sample complexity of the linear quadratic regulator.
\newblock {\em Foundations of Computational Mathematics}, 20(4):633--679, 2020.

\bibitem{sun2020finite}
Yue Sun, Samet Oymak, and Maryam Fazel.
\newblock Finite sample system identification: Optimal rates and the role of regularization.
\newblock In {\em Learning for dynamics and control}, pages 16--25. PMLR, 2020.

\bibitem{xing2020linear}
Yu~Xing, Ben Gravell, Xingkang He, Karl~Henrik Johansson, and Tyler Summers.
\newblock Linear system identification under multiplicative noise from multiple trajectory data.
\newblock In {\em 2020 American Control Conference (ACC)}, pages 5157--5261. IEEE, 2020.

\bibitem{xin2022learning}
Lei Xin, George Chiu, and Shreyas Sundaram.
\newblock Learning the dynamics of autonomous linear systems from multiple trajectories.
\newblock In {\em 2022 American Control Conference (ACC)}, pages 3955--3960. IEEE, 2022.

\bibitem{tu2024learning}
Stephen Tu, Roy Frostig, and Mahdi Soltanolkotabi.
\newblock Learning from many trajectories.
\newblock {\em Journal of Machine Learning Research}, 25(216):1--109, 2024.

\bibitem{rantzer2018}
Anders Rantzer.
\newblock Concentration bounds for single parameter adaptive control.
\newblock In {\em 2018 Annual American Control Conference (ACC)}, pages 1862--1866. IEEE, 2018.

\bibitem{tsiamis2019finite}
Anastasios Tsiamis and George~J Pappas.
\newblock Finite sample analysis of stochastic system identification.
\newblock In {\em 2019 IEEE 58th Conference on Decision and Control (CDC)}, pages 3648--3654. IEEE, 2019.

\bibitem{jedra2022finite}
Yassir Jedra and Alexandre Proutiere.
\newblock Finite-time identification of linear systems: Fundamental limits and optimal algorithms.
\newblock {\em IEEE Transactions on Automatic Control}, 68(5):2805--2820, 2022.

\bibitem{foster2020learning}
Dylan Foster, Tuhin Sarkar, and Alexander Rakhlin.
\newblock Learning nonlinear dynamical systems from a single trajectory.
\newblock In {\em Learning for Dynamics and Control}, pages 851--861. PMLR, 2020.

\bibitem{ziemann2022single}
Ingvar~M Ziemann, Henrik Sandberg, and Nikolai Matni.
\newblock Single trajectory nonparametric learning of nonlinear dynamics.
\newblock In {\em conference on Learning Theory}, pages 3333--3364. PMLR, 2022.

\bibitem{faradonbeh2018finite}
Mohamad Kazem~Shirani Faradonbeh, Ambuj Tewari, and George Michailidis.
\newblock Finite time identification in unstable linear systems.
\newblock {\em Automatica}, 96:342--353, 2018.

\bibitem{vladimirova2020sub}
Mariia Vladimirova, St{\'e}phane Girard, Hien Nguyen, and Julyan Arbel.
\newblock Sub-weibull distributions: Generalizing sub-gaussian and sub-exponential properties to heavier tailed distributions.
\newblock {\em Stat}, 9(1):e318, 2020.

\bibitem{kanakeri2024outlier}
Vinay Kanakeri and Aritra Mitra.
\newblock Outlier-robust linear system identification under heavy-tailed noise.
\newblock {\em arXiv preprint arXiv:2501.00421}, 2024.

\bibitem{tsiamis2021linear}
Anastasios Tsiamis and George~J Pappas.
\newblock Linear systems can be hard to learn.
\newblock In {\em 2021 60th IEEE Conference on Decision and Control (CDC)}, pages 2903--2910. IEEE, 2021.

\bibitem{horn2012matrix}
Roger~A Horn and Charles~R Johnson.
\newblock {\em Matrix analysis}.
\newblock Cambridge university press, 2012.

\end{thebibliography}

\appendix
\section{Preliminary results}\label{app:preliminary_results}
In the proofs of Lemmas \ref{lemma:invertibility}, \ref{lemma:process_noise}, and \ref{lemma:measurement_noise}, we bound the Frobenius norm of random matrices. In this regard, the following result defines a variance statistic, $\text{var}(\cdot)$, and shows how it exploits independence. 
\begin{lemma} \label{lemma:var_statistic}
    Let $\{X_i\}_{i \in [n]}$ be a collection of independent matrices. Define $$\text{var}(X_i) \triangleq \E[\normlr{X_i - \E[X_i]}_F^2] = \E[\norm{X_i}_F^2] - \norm{\E[X_i]}_F^2.$$ We have $$\text{var}\left(\sum_{i \in [n]} X_i \right) = \sum_{i \in [n]} \text{var}(X_i).$$
\end{lemma}
See Lemma 9 in \cite{kanakeri2024outlier} for proof. 

% \begin{lemma}{\textbf{(Markov's inequality for matrices):}}
% \label{lemma:matrix_markov}
%     Let $X$ be a random matrix such that $X \succeq 0$ almost surely and with expectation $\mathbb{E}[X]$, and let $A \succ 0$, then
%     \begin{equation}
%     \label{eqn:matrix_markov}
%         \mathbb{P}\left( X \not\preceq A \right) \leq \tr{\E[X]A^{-1}}. \nonumber
%     \end{equation}
% \end{lemma}
% For the proof of the above result, see Theorem 12 from \cite{ahlswede2002strong}. 
The following result captures the robustness of the geometric median and is used in the proof of Lemma \ref{lemma:boosting}.
\begin{lemma} \textbf{(Property of the geometric median):}
\label{lemma:geometric_median}
Let $A_1, \ldots, A_K \in \mathbb{R}^{m \times n}$ and let $A_*$ be their geometric median with respect to the Frobenius norm: $$A_* = \texttt{Med}(A_1, \ldots, A_K) := \argmin_{\theta \in \mathbb{R}^{m \times n}} \sum_{j \in [K]} \Vert \theta - A_j \Vert_F.$$ Fix $\alpha \in (0, 0.5)$ and assume $A \in \mathbb{R}^{m \times n}$ is such that $\normlr{A_* - A}_F > C_\alpha r$, where $C_\alpha = (1 - \alpha) \sqrt{1/(1 - 2\alpha)}$ and $r > 0$. Then,  there exists a subset $J \subseteq [K]$ of cardinality $\abs{J} > \alpha K$ such that for all $j \in J$, $\normlr{A_j - A}_F > r$.
\end{lemma}
See Lemma 2.1 in \cite{minsker2015geometric} for the proof of the above lemma.

\section{Proofs of Lemmas \ref{lemma:invertibility}-\ref{lemma:measurement_noise}} \label{app:proofs}
We begin with the proof of Lemma \ref{lemma:invertibility} which provides a lower bound for $\lambda_{\min}(U_jU_j^\top)$.

\subsection{Proof of Lemma \ref{lemma:invertibility}}
Defining 
\begin{equation*} 
    \begin{aligned}
    \hat{u}_t^{(i)} = \begin{bmatrix}
    (u_t^{(i)})^\top & (u_{t-1}^{(i)})^\top & \ldots & (u_0^{(i)})^\top & 0 & \ldots & 0
    \end{bmatrix}^\top \in \mathbb{R}^{mT},
    \end{aligned}
\end{equation*}
we have 
\begin{align*}
    U_jU_j^\top = \sum_{i \in \mc{B}_j} U^{(i)}(U^{(i)})^\top = \sum_{i \in \mc{B}_j} \sum_{t = 0}^{T-1}\hat{u}_t^{(i)}(\hat{u}_t^{(i)})^\top.
\end{align*}
Using the Weyl's inequality (Theorem 4.3.1 from \cite{horn2012matrix}), we have $\lambda_{\min}(U_jU_j^\top) \geq \sum_{t = 0}^{T-1}\lambda_{\min}(\sum_{i \in \mc{B}_j} \hat{u}_t^{(i)}(\hat{u}_t^{(i)})^\top)$. Notice that $\lambda_{\min}(\sum_{i \in \mc{B}_j} \hat{u}_t^{(i)}(\hat{u}_t^{(i)})^\top) \geq 0$, since $\sum_{i \in \mc{B}_j} \hat{u}_t^{(i)}(\hat{u}_t^{(i)})^\top \succeq 0$ for all $t \in \{0, 1, \ldots, T-1\}$. Therefore, we can further lower bound $\lambda_{\min}(U_jU_j^\top)$ by only considering the last index, $t = T-1$, in the summation. Hence, we obtain $\lambda_{\min}(U_jU_j^\top) \geq \lambda_{\min}(\sum_{i \in \mc{B}_j} \hat{u}_{T-1}^{(i)}(\hat{u}_{T-1}^{(i)})^\top)$. We choose the index $T-1$ because from the definition of $\hat{u}_t^{(i)}$, $\lambda_{\min}(\hat{u}_t^{(i)}(\hat{u}_t^{(i)})^\top) = 0$ for all $t \in \{0, 1, \ldots, T-2\}$, which does not lead to a meaningful lower bound. 

To bound $\lambda_{\min}(\sum_{i \in \mc{B}_j} \hat{u}_{T-1}^{(i)}(\hat{u}_{T-1}^{(i)})^\top)$, consider the following:
\begin{align}
    \sum_{i \in \mc{B}_j} \hat{u}_{T-1}^{(i)}(\hat{u}_{T-1}^{(i)})^\top &= \E\left[\sum_{i \in \mc{B}_j} \hat{u}_{T-1}^{(i)}(\hat{u}_{T-1}^{(i)})^\top \right] + X \nonumber\\
    &= M \sigma_u^2 I_{mT} + X, \label{eqn:den_1}
\end{align}
where we used the fact that there are $M$ i.i.d rollouts in $\mc{B}_j$ and $\E[\hat{u}_{T-1}^{(i)}(\hat{u}_{T-1}^{(i)})^\top] = \sigma_u^2I_{mT}$, and defined $X := \sum_{i \in \mc{B}_j} \hat{u}_{T-1}^{(i)}(\hat{u}_{T-1}^{(i)})^\top -  M \sigma_u^2 I_{mT}$. Using the Markov's inequality on $\norm{X}_F^2$, we have for any $\gamma > 0$
\begin{align}
    \mathbb{P}(\norm{X}_F^2 \geq \gamma^2) &\leq (1/\gamma^2)\E[\norm{X}_F^2] \nonumber\\
    & \overset{(a)}{=} (1/\gamma^2) M\E[\norm{ \hat{u}_{T-1}^{(1)}(\hat{u}_{T-1}^{(1)})^\top - \sigma_u^2 I_{mT}}_F^2] \nonumber\\
    & = (1/\gamma^2) M \E[\norm{\hat{u}_{T-1}^{(1)}(\hat{u}_{T-1}^{(1)})^\top}_F^2] - \sigma_u^4 mT \nonumber \\
    & \overset{(b)}{\leq} (1/\gamma^2) M \E[\norm{\hat{u}_{T-1}^{(1)}(\hat{u}_{T-1}^{(1)})^\top}_F^2] \nonumber\\
    & = (1/\gamma^2) M \E[(\norm{\hat{u}_{T-1}^{(1)}}^2)^2] \nonumber\\
    & \overset{(c)}{\leq} (1/\gamma^2) M (mT)^2 \tilde{\sigma}_u^4, \label{eqn:den_2}
\end{align}
where $(a)$ follows from Lemma \ref{lemma:var_statistic} by observing that $X$ is a sum of i.i.d matrices, $(b)$ by dropping the term $\sigma_u^4 mT$, and $(c)$ by observing that 
\begin{align*}
    \E[\norm{\hat{u}_{T-1}^{(1)}}^4] &= \E\left[\sum_{i = 1}^{mT} (\hat{u}_{T-1}^{(1)}(i))^4 \right] \\
    & \quad +\E\left[\sum_{i \neq l = 1}^{mT} (\hat{u}_{T-1}^{(1)}(i))^2 (\hat{u}_{T-1}^{(1)}(l))^2 \right] \\
    &\overset{(d)}{=} mT\tilde{\sigma}_u^4 + ((mT)^2 - mT)\sigma_u^4 \\
    &\overset{(e)}{\leq} (mT)^2 \tilde{\sigma}_u^4,
\end{align*}
where $(d)$ follows based on the noise assumptions \eqref{eqn:noise_model}, and $(e)$ as $\tilde{\sigma}_u^4 \geq \sigma_u^4$ due to Jensen's inequality. Choosing $\gamma = M\sigma_u^2/2$ and setting the R.H.S $\leq q/3$ in \eqref{eqn:den_2}, we obtain the bound for the number of rollouts $M \geq 12(1/q)(mT)^2 (\tilde{\sigma}_u^4/\sigma_u^4)$. On the successful event, with probability at least $1 - q/3$, $\norm{X}_F \leq  M\sigma_u^2/2$, which implies $\norm{X} \leq  M\sigma_u^2/2$. Since spectral radius is less than the spectral norm, we have $\lambda_{\min}(X) \geq - M\sigma_u^2/2$. Finally, based on \eqref{eqn:den_1}, $\lambda_{\min}(\sum_{i \in \mc{B}_j} \hat{u}_{T-1}^{(i)}(\hat{u}_{T-1}^{(i)})^\top) \geq M \sigma_u^2/2$.

\subsubsection{Proof of Lemma \ref{lemma:process_noise}} 
Following an approach similar to the proof of Proposition 3.3 in \cite{zheng2021nonasymptotic}, consider the $(k, l)$-th block, $(W_jU_j^\top)_{k,l} \in \mathbb{R}^{n \times m}$, of the matrix $W_jU_j^\top$. We have $\forall k \in [T], \forall l \in [T]$ $$(W_jU_j^\top)_{k,l} = \sum_{i \in \mc{B}_j} \sum_{t = 0}^{T-k} w_t^{(i)}(u_{t+k-l}^{(i)})^\top.$$ We proceed by bounding the norm of each block. Defining $X_{k, l} =(W_jU_j^\top)_{k,l}$ and $T_{k, l} := T - \max(k, l) + 1$, we have
\begin{align}
    \mathbb{P}(\norm{X_{k, l}}_F^2 \geq \gamma^2) &\leq (1/\gamma^2)\E[\norm{X_{k, l}}_F^2] \nonumber\\
    & \overset{(a)}{=} (1/\gamma^2) MT_{k, l}\E[\norm{w_0^{(1)}(u_0^{(1)})^\top}_F^2] \nonumber \\
    & \overset{(b)}{=} (1/\gamma^2) MT_{k, l}\E[\norm{w_0^{(1)}}^2]\E[\norm{u_0^{(1)}}^2] \nonumber \\
    & \overset{}{=} (1/\gamma^2) MT_{k, l} \sigma_u^2\sigma_w^2 nm \nonumber \\
    & \overset{(c)}{\leq} (1/\gamma^2) MT \sigma_u^2\sigma_w^2 nm, \label{eqn:num_1}
\end{align}
where $(a)$ follows from Lemma \ref{lemma:var_statistic} and observing that there are $MT_{k, l}$ i.i.d copies of $w_0^{(1)}(u_0^{(1)})^\top$ in $X_{k, l}$, and that $w_0^{(1)}$ and $u_0^{(1)}$ are independent and zero mean random vectors; $(b)$ due to the following: $\norm{w_0^{(1)}(u_0^{(1)})^\top}_F^2 = \tr{w_0^{(1)}(u_0^{(1)})^\top u_0^{(1)}(w_0^{(1)})^\top} = \norm{u_0^{(1)}}^2\norm{w_0^{(1)}}^2$. Finally, $(c)$ follows as $T_{k, l} \leq T, \forall k \in [T], \forall l \in [T]$. Setting the R.H.S $= q/(3T^2)$ in \eqref{eqn:num_1}, we get with probability at least $1 - q/(3T^2)$, $\norm{X_{k, l}}_F^2 \leq 3(1/q)T^3\sigma_u^2\sigma_w^2Mnm$. Applying an union bound over the $T^2$ blocks in $W_jU_j^\top$, with probability at least $1 - q/3$, $\norm{X_{k, l}}_F^2 \leq  3(1/q)T^3\sigma_u^2\sigma_w^2Mnm$, $\forall k \in [T], \forall l \in [T]$. To transfer the bound from the block matrices to the complete matrix $W_jU_j^\top$, notice that $\norm{W_jU_j^\top}_F^2 = \sum_{k \in [T]}\sum_{l \in [T]} \norm{X_{k, l}}_F^2$. Hence, we have $\norm{W_jU_j^\top}_F^2 \leq  3(1/q)T^5\sigma_u^2\sigma_w^2Mnm$. Finally, the proof is complete by using $\norm{W_jU_j^\top} \leq \norm{W_jU_j^\top}_F$.

\subsubsection{Proof of Lemma \ref{lemma:measurement_noise}}
The matrix $V_jU_j^\top \in \mathbb{R}^{p \times mT}$ can be represented as
\begin{equation*}
    \begin{aligned}
    V_jU_j^\top = \sum_{i \in \mc{B}_j}\begin{bmatrix}
        \sum_{t = 0}^{T-1}v_t^{(i)}(u_t^{(i)})^\top& \sum_{t = 0}^{T-2}v_{t+1}^{(i)}(u_t^{(i)})^\top & \ldots & v_{T-1}^{(i)}(u_0^{(i)})^\top
    \end{bmatrix}.
    \end{aligned}
\end{equation*}
We first bound the $p \times m$ block matrices $\sum_{i \in \mc{B}_j}\sum_{k=0}^{T-l-1}v_{k+l}^{(i)}(u_k^{(i)})^\top$ for all $l \in \{0, 1, \ldots, T-1\}$, and transfer the bounds to the complete matrix $V_jU_j^\top$. To bound the block matrices, we use the Markov's inequality and follow the analysis similar to the proof of Lemma \ref{lemma:process_noise}. Defining $X_l:= \sum_{i \in \mc{B}_j}\sum_{k=0}^{T-l-1}v_{k+l}^{(i)}(u_k^{(i)})^\top$, we have 
\begin{align}
    \mathbb{P}(\norm{X_l}_F^2 \geq \gamma^2) &\leq (1/\gamma^2)\E[\norm{X_l}_F^2] \nonumber \\
    &= (1/\gamma^2)M(T-l)\E[\norm{v_{0}^{(1)}(u_0^{(1)})^\top}_F^2] \nonumber \\
    &= (1/\gamma^2)M(T-l)\E[\norm{v_0^{(1)}}^2]\E[\norm{u_0^{(1)}}^2] \nonumber \\
    &= (1/\gamma^2)M(T-l)\sigma_u^2\sigma_v^2pm \nonumber \\
    &\leq (1/\gamma^2)MT\sigma_u^2\sigma_v^2pm, \label{eqn:num_2}
\end{align}
where the first step follows from Lemma \ref{lemma:var_statistic} and observing that there are $M(T-l)$ i.i.d copies of $v_{0}^{(1)}(u_0^{(1)})^\top$ in $X_l$, and that $v_0^{(1)}$ and $u_0^{(1)}$ are independent and zero mean random vectors. The remaining steps follow based on arguments made in the proof of Lemma \ref{lemma:process_noise}. Setting the R.H.S = $q/(3T)$ in \eqref{eqn:num_2}, we have with probability at least $1 - q/(3T)$, $\norm{X_l}_F^2 \leq 3(1/q)T^2\sigma_v^2 \sigma_u^2Mpm$. Applying an union bound over the $T$ blocks in $V_jU_j^\top$, with probability at least $1 - q/3$, $\norm{X_l}_F^2 \leq  3(1/q)T^2\sigma_v^2 \sigma_u^2Mpm$, for all $l \in \{0, 1, \ldots, T-1\}$. Finally, since $\norm{V_jU_j^\top}_F^2 = \sum_{l = 0}^{T-1} \norm{X_l}_F^2$, we have $\norm{V_jU_j^\top}_F^2 \leq 3(1/q)T^3\sigma_v^2 \sigma_u^2Mpm$. The proof is complete by using $\norm{V_jU_j^\top} \leq \norm{V_jU_j^\top}_F$.

\end{document}